\documentclass[letterpaper, 10 pt, conference]{ieeeconf}

\IEEEoverridecommandlockouts                              
\usepackage{graphics} 
\usepackage{epsfig} 
\usepackage{times} 
\usepackage{amsmath} 
\usepackage{amssymb}  
\usepackage{subfigure}
\usepackage{cite}
\usepackage{amsfonts}
\usepackage{fancybox}
\usepackage{relsize}
\usepackage{floatrow}
\usepackage{colortbl}
\usepackage{hhline}
\usepackage{lipsum}
\usepackage{mathtools}
\usepackage{cuted}
\usepackage{subfig}
\usepackage{romannum}
\usepackage{dsfont}
\usepackage{algorithm}
\usepackage{algorithmic}
\usepackage{figsize}

\newtheorem{thm1}{\bf Theorem}
\newtheorem{prop1}{\bf Proposition}
\newtheorem{lem1}{\bf Lemma}
\newtheorem{assmpt1}{\bf Assumption}
\newtheorem{defn17}{\bf Definition}
\newtheorem{rem1}{\bf Remark}

\newtheorem{cor1}{\bf Corollary}

\newenvironment{assumption}{\begin{assmpt1}}{\hfill$\Diamond$\end{assmpt1}}
\newenvironment{definition}{\begin{defn17}}{\hfill$\Diamond$\end{defn17}}
\newenvironment{remark}{\begin{rem1}}{\hfill$\Diamond$\end{rem1}}
\newenvironment{lemma}{\begin{lem1}}{\hfill$\Diamond$\end{lem1}}
\newenvironment{theorem}{\begin{thm1}}{\hfill$\Diamond$\end{thm1}}

\title{\LARGE \bf
Robust Control for Lane Keeping System Using Linear Parameter Varying Approach  with Scheduling Variables Reduction
}

\author{Ying Shuai Quan$^{1}$, Jin Sung Kim$^{1}$ and Chung Choo Chung$^{2}$$^\dag$
\thanks{$^{1}$Y. S. Quan and J. S. Kim are with Dept. of Electrical Engineering, Hanyang University, Seoul 04763, Korea. (e-mail: yeongsu.quan@gmail.com, jskim06@hanyang.ac.kr)}
\thanks{$^{2}$C. C. Chung is with Div. of Electrical and Biomedical Engineering, Hanyang University, Seoul 04763, Korea. (+82-2-2220-1724, e-mail: cchung@hanyang.ac.kr)}
\thanks{\dag: Corresponding author 
}
}

\begin{document}

\maketitle
\thispagestyle{empty}
\pagestyle{empty}

\begin{abstract}

This paper presents a robust controller using a  Linear Parameter Varying (LPV) model of the lane-keeping system with parameter reduction. Both varying vehicle speed and roll motion on a curved road influence the lateral vehicle model's parameters, such as tire cornering stiffness. Thus, we use the LPV technique to take the parameter variations into account in vehicle dynamics. However, multiple varying parameters lead to a high number of scheduling variables and cause massive computational complexity. In this paper, to reduce the computational complexity, Principal Component Analysis (PCA)-based parameter reduction is performed to obtain a reduced model with a tighter convex set. We designed the LPV robust feedback controller using the reduced model solving a set of Linear Matrix Inequality (LMI). The effectiveness of the proposed system is validated with full vehicle dynamics from CarSim on an interchange road. From the simulation, we confirmed that the proposed method largely reduces the lateral offset error, compared with other controllers based on Linear Time-Invariant (LTI) system.

\end{abstract}

\section*{NOMENCLATURE}
\begin{itemize}
\item $C_{\alpha i}$ : Cornering stiffness of tire, $i \in \{f, r\}$
\item $V_x$ : Longitudinal speed
\item $m$ : Total mass of vehicle
\item $l_i$ : Distance between front(rear) tire and center of gravity(CG) , $i \in \{f, r\}$
\item $I_z$ : Yaw moment of inertia of vehicle
\item $L$ : Look-ahead distance
\item $e_y=y-y_{des}$ : Lateral position error in local coordinate w.r.t. lane
\item $\psi$ : Yaw angle of vehicle in global coordinate
\item $e_{\psi}=\psi_{des}-\psi$ : Heading angle error in local coordinate w.r.t. lane
\end{itemize}
\section*{Parameters}
\begin{gather*}
a_{22}=-\frac{2C_{\alpha f}+2C_{\alpha r}}{mV_{x}},~  a_{23}=-a_{22}V_{x},\\
a_{24}=-1-\frac{2C_{\alpha f}l_{f}-2C_{\alpha r}l_{r}}{m{V_{x}}^2},~ a'_{24}=(a_{24}-1)V_{x},\\
a_{42}=-\frac{2C_{\alpha f}l_{f}-2C_{\alpha r}l_{r}}{I_{z}},~ a'_{42}=a_{42}/V_{x},\\
a_{43}=-a_{42},~ a_{44}=-\frac{2C_{\alpha f}{l_{f}}^{2}+2C_{\alpha r}{l_{r}}^{2}}{I_{z}V_{x}},
\label{dynamic system co.}
\end{gather*}
\section{INTRODUCTION}
For autonomous driving vehicles, lateral dynamic motion control and longitudinal velocity control are needed to prevent rollover on a curved road.
There are researches on designing the vehicle controller and estimator for rollover prevention~\cite{chen2001differential,yoon2008unified}.
One effective maneuvering method is to reduce the longitudinal speed in front of a curved road because it highly influences vehicle roll motion~\cite{rajamani2011vehicle}.
Therefore, proper speed control is considered necessary on curved roads or interchanges for driving stability and comfort.
Varying vehicle speed leads to tire vertical load variation.
Literature shows that vehicle roll motion leads to the lateral load transfer, which also influences the tire vertical load~\cite{yoon2008unified}.
Since cornering stiffness is mostly affected by vehicle vertical load~\cite{kim2019vehicular}, cornering stiffness variation becomes non-ignorable in vehicle lateral control on the curved road.
For the varying cornering stiffness, the estimation method has been actively studied.
There are plenty of papers showing that cornering stiffness can be estimated.
In~\cite{sierra2006cornering}, a vehicle yaw/lateral model based on a single-track model was exploited, and methods for cornering stiffness estimation were presented.
%
A method estimating cornering stiffness and tire-road friction was studied where the vehicle was running in a straight road and made a fast turn~\cite{wang2013tire}.
In a recent study, tire stiffness and the vehicle state were estimated with Bayesian framework~\cite{berntorp2018tire}.

Considering the model variation that vehicle speed and cornering stiffness bring to the lateral dynamic motion model, gain scheduling based on Linear Parameter Varying (LPV) system models could be an effective and reliable method for vehicle lateral control.
There is literature where only speed variation or cornering stiffness is considered in the robust lateral control design using LPV models~\cite{lee2020autonomous, li2014lpv}.
At the same time, few studies have been reported with variations of both considered.
However, investigating variations of both vehicle speed and cornering stiffness leads to a higher number of scheduling variables.
For polytopic LPV models, the complexity of the controller synthesis procedure grows exponentially with the number of scheduling variables~\cite{kwiatkowski2008pca}.
Both online membership determination and offline analysis for local controllers become computationally expensive and intractable~\cite{kwiatkowski2008pca}.
Moreover, the convex polytope containing the scheduling variable trajectory is built, assuming that all parameters vary independently while they are inherently coupled with each other.
It is well-reported in~\cite{kwiatkowski2008pca} that the resulting convex parameter set might be loose and conservative, and includes vertices even unreasonable for the real plant.

To consider the lateral vehicle dynamics with multiple varying parameters, it is crucial to reduce the scheduling variable dimension for computational simplicity.
With such numerical reductions, the amount of LMI constraints, decision variables, online computational load and hardware resources requirements can be significantly reduced~\cite{lipolytopic}.
There exist several studies where the dimension reduction is conducted by data-based methods, such as Principal Component Analysis (PCA)~\cite{kwiatkowski2008pca, hashemi2009lpv}, Kernel PCA (KPCA)~\cite{rizvi2016kernel}, Autoencoders (AE)~\cite{rizvi2018model} and Deep Neural Network (DNN)~\cite{koelewijn2020scheduling}.
Both KPCA and AE are nonlinear mapping with an extra optimization procedure to ensure the affine dependence of reduced model matrices on the new scheduling variables. In contrast, the DNN method extracts the system matrices of the reduced model by adding additional hidden layers.
PCA method extracts the most significant principal components of the collected scheduling variable trajectories by Singular Value Decomposition (SVD).
Since the PCA method is linear mapping, the reduced LPV system matrices can be directly constructed by the inverse linear mapping~\cite{jackson2005user} with affine dependence satisfied.

This paper presents a robust controller using an LPV model of the lane-keeping system with parameter reduction. On interchange roads, we could ignore neither varying vehicle longitudinal speed nor cornering stiffness variations.
Thus, we consider both vehicle speed and cornering stiffness as varying parameters in the lateral model, leading to the problem that finding a vertex membership is computationally expensive with the resulting convex polytope over conservative.
We resolve it using the dimension reduction procedure to the scheduling variables. The PCA method is chosen for its advantages of linear property and computational simplicity.
We designed the LPV robust feedback controller with the reduced model solving a set of Linear Matrix Inequality (LMI).
The reconstruction error of applying PCA is considered in the local controllers design with a tightened Lyapunov function to maintain robust performance with estimation error.
The proposed system's effectiveness is validated by the full vehicle
dynamic simulator using CarSim with MATLAB/Simulink. From simulation results, we confirmed that the proposed method reduces the lateral offset error by about $30$\%, compared with other controllers based on Linear Time-Invariant (LTI) system.
\section{Vehicle Modeling}
In this section, vehicle models are derived for representing the roll motion and the lateral dynamic motion.
Notice that the roll dynamics is for the necessary speed control in the curved road which leads to the speed variation later considered in vehicle lateral control design.
First, simple one degree-of-freedom(DOF) roll dynamics is adopted to estimate the roll angle~\cite{rajamani2011vehicle, yoon2008unified} to avoid the complicated suspension dynamics.
During cornering maneuvers, vehicle roll motion is caused by lateral accelerations and can be described by a roll model involving the roll angle($\phi$). The model of the roll dynamics could be presented by
\begin{equation}
\ddot{\phi} = \frac{m_s h_{rc}(a_y + g \cdot \text{sin}\phi)}{I_x} - \frac{K_{roll}}{I_x} - \frac{C_{roll}}{I_x} \dot{\phi}
\label{eq:roll model}
\end{equation}
where $m_s$ is sprung mass, $h_{rc}$ is height of the roll center from the c.g., $a_y$ is lateral acceleration, $g$ is acceleration of gravity, $K_{roll}$ is roll stiffness, and $C_{roll}$ is roll damping coefficient~\cite{rajamani2011vehicle}.
\begin{assumption}
\label{asum:roll angle}
The estimator for roll motion is designed such that the roll angle of the vehicle body is available.
\end{assumption}
Wirh estimation methods~(see e.g., \cite{kim2019vehicular, yoon2008unified}), Assumption~\ref{asum:roll angle} is reasonable.
In steady-state cornering motion, the roll rate($\dot{\phi}$) and acceleration($\ddot{\phi}$) can be considered zero.
Using the approximation that $a_y = \ddot{y}+V_x \dot{\psi} \approx {V_{x}}^2 /R$~\cite{rajamani2011vehicle}, the desired longitudinal speed is presented as
\begin{equation}
V_x^d \leq  \sqrt{\frac{R (K_{roll}-m_s g h_{rc})}{m_s h_{rc}} \phi_{max} }
\label{eq:desired speed}
\end{equation}
with the condition of $| \phi_s | \leq | \phi_{max} |$, where $R$ is the turning radius of road.

We then derive the lateral dynamic motion model where the vehicle speed and the cornering stiffness are considered as varying parameters, which is a suitable choice in high speed or dynamic scenario~\cite{son2014robust}. In the dynamic situation, the direction of the tire is no longer equal to the direction of the velocity at each wheel. Thus, there exist tire slip angles that result in the lateral tire forces for the vehicle wheels, which can be linearly approximated as functions of the tire slip angles~\cite{rajamani2011vehicle}.
%
%
%
Then, the lateral dynamic motion model is derived in terms of error with respect to ego lane. The state of the model is given as $\textbf{x} =
\begin{bmatrix}
   e_{yL} & \dot{e}_{y} & e_{\psi} & \dot{\psi}
\end{bmatrix}^{T} \in \mathds{R}^{4}$.
For $e_{yL}$, $\dot{e}_{y}$, and $e_{\psi}$, the camera vision system detecting the lane mark is used.
From the lane mark, the cubic-polynomial road lane model is calculated to get the state.
The state $\dot{\psi}$ is measurable via the vehicle Inertia Measurement Unit (IMU) sensor.
Consequently, the lateral dynamic model is described as follows~(see the author's paper~\cite{lee2016robust} for details):
\begin{equation}
\begin{split}
\dot{\textbf{x}}&=A{\textbf{x}}+B \textbf{u}+B_{\varphi}\varphi
\end{split}
\label{eq:continuous-time model}
\end{equation}
where
\begin{gather*}
A=\begin{bmatrix} 0&1&0&-L\\
                 0&a_{22}&a_{23}&a'_{24}\\
                 0&0&0&-1\\
                 0&a'_{42}&a_{43}&a_{44}\end{bmatrix},~
                 B=
                 \begin{bmatrix}
                 0\\\frac{2C_{\alpha f}}{m} \\0\\\frac{2C_{\alpha f}l_{f}}{I_{z}}
                 \end{bmatrix},\\
                 B_{\varphi}=
                 \begin{bmatrix}
                 L&V_{x}\\V_{x}&0\\1&0\\0&0
                 \end{bmatrix},
~\textbf{u}=\delta ,~ \varphi =\begin{bmatrix} {\dot{\psi}}_{des}\\
                                    e_{\psi L}-e_{\psi}\end{bmatrix}.
\end{gather*}
%
%
where $\delta$ is steering angle.
In this paper, it is assumed that the system state is always available with appropriate sensors and estimation methods.
\section{Linear Parameter Varying System and Lateral Control}

\subsection{Varying Parameter}
In this paper, longitudinal speed and cornering stiffness variations are considered in the lateral dynamic model.
We assume that vehicle speed is measurable by the in-vehicle sensors, and cornering stiffness is estimated~\cite{wang2013tire}.
Note that here we only consider the lane-keeping problem on dry asphalt roads, where cornering stiffness can be linearly estimated and always available~\cite{son2014robust,lee2016robust}.

\begin{assumption}
The tire cornering stiffness is estimated with the estimation method.
\end{assumption}

Considering the impact that the varying parameters $V_x$, $C_{\alpha f}$ and $C_{\alpha r}$ bring to the lateral dynamic model, one can define the scheduling variables as:
\begin{align*}
\boldsymbol{\theta}= \begin{bmatrix} \theta_1 & \theta_2& \theta_3& \theta_4 & \theta_5  \end{bmatrix}^T
\end{align*}
where
%
%
$
\theta_1= { V_x},~
\theta_2=2C_{\alpha f},~
\theta_3= \frac{2C_{\alpha f}}{ V_x},~
\theta_4=2C_{\alpha r},~ \text{and}~
\theta_5= \frac{2C_{\alpha_r}}{ V_x}$.
%
%
%
To ensure the affine dependency of LPV systems on $\boldsymbol{\theta}$, $\theta_3$ and $\theta_5$ are introduced for the nonlinear relationships between $C_{\alpha f}$, $C_{\alpha r}$ and $V_x$.
With scheduling variable $\boldsymbol{\theta}$, the system matrix in~(\ref{eq:continuous-time model}) can be rewritten as:
\begin{equation}
\begin{split}
\dot{\textbf{x}}&=A(\boldsymbol{\theta}){\textbf{x}}+B(\boldsymbol{\theta}) \textbf{u}+B_{\varphi}(\boldsymbol{\theta})\varphi
\end{split}
\label{eq:continuous-time model2}
\end{equation}
where
\begin{align*}
A(\boldsymbol{\theta})= ~~~~~~~~~~~~~~~~~~~~~~~~~~~~~~~~~~~~~~~~~~~~~~~~~~~~~~~~~~~~~~\\
\begin{bmatrix} 0 & 1 & 0 & -L  \\
 0 & -\frac{1}{ m}\theta_3-\frac{1}{ m}\theta_5
  &\frac{1}{m}\theta_2+ \frac{1}{ m}\theta_4
   & -2\theta_1-\frac{l_f}{ m}\theta_3 + \frac{l_r}{ m}\theta_5  \\
    0 & 0 & 0 & -1 \\
     0 & -\frac{l_f}{ I_z }\theta_3 + \frac{l_r}{I_z}\theta_5
      & \frac{l_f}{I_z}\theta_2-\frac{l_r}{ I_z}\theta_4
      & -\frac{{l_f}^2}{ I_z}\theta_3-\frac{{l_r}^2}{I_z}\theta_5
     \\
        \end{bmatrix},
\end{align*}

           \begin{align*}
      B(\boldsymbol{\theta})=\begin{bmatrix} 0 \\ \frac{1  }{m}\theta_2 \\ 0 \\ \frac{ l_f }{I_z}  \theta_2 \end{bmatrix},~~
      B_{\varphi}(\boldsymbol{\theta}) = \begin{bmatrix} L& \theta_1 \\ \theta_1 &0\\1&0\\0&0\end{bmatrix}.
      \end{align*}

%

\subsection{Model Reduction via PCA}
For polytopic LPV models, the complexity of controller design has an exponential dependence on the number $l$ of scheduling variables.
Not only for the fact that the number of LMIs to solve offline grows exponentially with $l$, but online membership determination is also computationally intractable with $l$ larger than three ~\cite{kwiatkowski2008pca}.
Here a PCA algorithm is used to reduce the scheduling variable dimension.

To perform the dimension reduction procedure, the scheduling variable $\boldsymbol{\theta}(t)$ is firstly discretely sampled and collected during simulation.
%
%
The sampled trajectory $\Theta$ of $\boldsymbol{\theta}(t)$ is presented as a $l\times N$ matrix:
\begin{equation}
\Theta = \begin{bmatrix} \boldsymbol{\theta}(0) & ... &  \boldsymbol{\theta}((N-1)T) \end{bmatrix}
\end{equation}
with $N$ number of samples.
The rows of $\Theta_i$ of the trajectory matrix $\Theta$ are then normalized by an affine law $\mathcal{N}_i$ respectively, to obtain normalized trajectory matrix $\Theta^n=\mathcal{N}(\Theta)$ with
\begin{equation}
\Theta^n_i = \mathcal{N}_i(\Theta_i), \Theta_i = \mathcal{N}_i^{-1} (\Theta^n_i), i=1,...,l.
\end{equation}
Then the PCA algorithm is applied to the normalized trajectory by introducing a SVD~\cite{jackson2005user}
\begin{equation}
\Theta^n =  \begin{bmatrix} U_s & U_n  \end{bmatrix}
\begin{bmatrix} \Sigma_s & 0 & 0 \\ 0 & \Sigma_n & 0  \end{bmatrix}
\begin{bmatrix} V_s^T \\ V_n^T  \end{bmatrix}
\end{equation}
where $U_s$, $\Sigma_s$ and $V_s$ are corresponded to the $m<l$ number of most significant singular values, $U_s \in \mathds{R}^{l \times m}$ is the basis of the significant column space of the scheduling variable trajectory $\Theta^n$, and $\Sigma_s V_s^T$ represents the principal components of $\Theta_n$.
%
%
Then, we see that $ \Theta^n$ can be approximated by  $\hat{\Theta}^n =  U_s \Sigma_s V_s^T $ such as $ \Theta^n \cong \hat{\Theta}^n$.
Thus, the original LPV system, ${G}(\boldsymbol{\theta})$   for  $\boldsymbol{\theta}$ given by
 \begin{align*}
 {G}(\boldsymbol{\theta}):=
  \begin{bmatrix}
    \begin{array}{c|c} {A}( {\boldsymbol{\theta}}(t))
        &  {B}({\boldsymbol{\theta}}(t)) ~
           {B}_\varphi({\boldsymbol{\theta}}(t))
    \end{array}
  \end{bmatrix}
\end{align*}
can be approximated through the PCA with reduced scheduling variable $\boldsymbol{\eta}(t)$  given by
\begin{equation}
\boldsymbol{\eta} (t) =  U_s^T \mathcal{N}(\boldsymbol{\theta}(t)) \in \mathbb{R}^{m \times 1}
\label{eq:phi}
\end{equation}
in the form of
\begin{align*}
\hat{G}(\boldsymbol{\eta}):=&\begin{bmatrix}  \begin{array}{c|c} \hat{A}({\boldsymbol{\eta}}(t)) &  \hat{B}({\boldsymbol{\eta}}(t)) ~
\hat{B}_\varphi({\boldsymbol{\eta}}(t))
\end{array}   \end{bmatrix} .
\end{align*}
Then, with approximation of the corresponding original scheduling variable $\hat{\boldsymbol{\theta}}(t)$
\begin{equation}
\hat{\boldsymbol{\theta}}(t) = \mathcal{N}^{-1}(U_s \boldsymbol{\eta}(t))=\mathcal{N}^{-1}(U_s U_s^T \mathcal{N}(\boldsymbol{\theta}(t)))
  \label{eq:remodel}
\end{equation}
where $\mathcal{N}^{-1} $ stands for the row-wise re-scaling corresponded to $\mathcal{N}$, we can get the following inversely mapped model
\begin{align}
  \label{eq:remodel_inverse}
   G(\hat{\boldsymbol{\theta}}):= \begin{bmatrix}  \begin{array}{c|c} {A}(\hat{\boldsymbol{\theta}}(t)) &  {B}(\hat{\boldsymbol{\theta}}(t)) ~  {B}_\varphi(\hat{\boldsymbol{\theta}}(t))
  \end{array}   \end{bmatrix}.
\end{align}
In a summary, given $\boldsymbol{\theta}(t)$ at the current time, the reduced scheduling variable $\boldsymbol{\eta}(t)$ can be approximated by~(\ref{eq:phi}) and the corresponding LPV model~(\ref{eq:remodel_inverse}) can be generated by~(\ref{eq:remodel}).

\subsection{Vertex Membership }
After the dimension reduction by applying the PCA algorithm, the virtual trajectory of the $m$-dimension scheduling variables can be obtained by
 \begin{equation}
 H =  U_s^T \Theta^n \in \mathbb{R}^{m \times N}.
\end{equation}
For $m$-dimension scheduling variables, every $m+1$ vertices can be considered sufficient to build a convex hull.
The $m+1$ vertices, $\boldsymbol{\eta}_{v_{p}}^H$  for $p=1, \cdots, m+1$ selected from $M=2^m$ vertices candidates can be generated with the lower and upper bounds
 \begin{align*}
 \begin{split}
 \underline{\eta}_i^{H} = \min_{j}{H(i,j)},~
 \overline{\eta}_i^{H} = \max_{j}{H(i,j)}\\
 \end{split}
\end{align*}
where $i=1,...,m$ and $j = 1, ..., N$,  of the reduced scheduling variable trajectory $H$.
In other words, we construct each component of $\boldsymbol{\eta}_{v_{p}}^H$  selecting either $\overline{\eta}_i^{H}$ or $\underline{\eta}_i^{H}$ depending on the dimension, $m$.
Given the convex hull co$\mathcal{V}$ with the vertex matrix
\begin{align*}
\mathcal{V} = \begin{bmatrix}  \boldsymbol{\eta}_{v_1}^{H} & ... & \boldsymbol{\eta}_{v_{m+1}}^{H}   \end{bmatrix} \in \mathbb{R}^{m \times (m+1)},
\end{align*}
with current $\boldsymbol{\eta}(t) \in$ co$\mathcal{V} $, one can find the convex coordinate
\begin{equation}
\boldsymbol{\xi}(t) = \begin{bmatrix} \mathcal{V}\\{1}_{1 \times (m+1)}  \end{bmatrix} ^{-1}
\begin{bmatrix} \boldsymbol{\eta}(t)\\1  \end{bmatrix} \in \mathbb{R}^{(m+1) \times 1}
\label{eq:xi}
\end{equation}
such that $\boldsymbol{\eta}(t) = \mathcal{V} \boldsymbol{\xi}(t) $, with $\xi_p(t) \geq 0 $ and $\sum_{p=1}^{m+1} \xi_p(t)=1$.
The affine system can thus be represented by a linear combination of LTI models obtained by evaluating the system model at the vertices
\begin{equation}
{G}(\boldsymbol{\theta}(t)) \cong \hat{G}(\boldsymbol{\eta}(t)) \cong {G}(\boldsymbol{\hat\theta}(t))
=\sum_{p=1}^{m+1}\xi_p(t) {G}(\boldsymbol{\hat\theta}_{v_p}^{H})
\label{eq:lpvsys}
\end{equation}
where  $\boldsymbol{\hat{\theta}}_{v_p}^{H}$ is reconstructed  with  $\boldsymbol\eta_{v_p}^{H}$ from   (\ref{eq:remodel}).
Remind that the reduced $m$-dimension scheduling parameter $\eta$ does not have physical meaning and is used to determine the membership of the current scheduling parameter.

\begin{figure}
\includegraphics[width=\hsize]{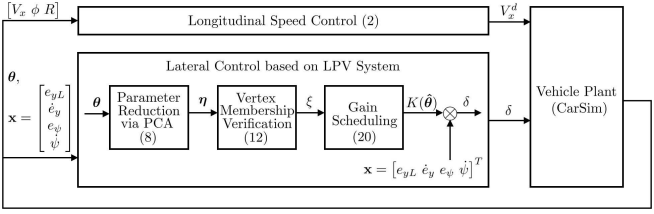}
\caption{Block diagram of the simulation}
\label{fig:simulation}
\end{figure}
\subsection{LPV Robust State Feedback Controller and Its Performance Analysis}
The reconstructed system model $G(\boldsymbol{\hat\theta})$ through PCA is an approximate form of the original system model $G(\boldsymbol{\theta})$. In this case, we assume that it is possible to find the upper bound of the model error~\cite{reiss2020nonasymptotic,milbradt2020high}.
\begin{assumption}
The error between the original parameter $\boldsymbol{\theta}$ and the reconstructed parameter $\boldsymbol{\hat\theta}$ is bounded.
\label{asum: PCA error}
\end{assumption}
Given $\epsilon>0$, we define ${B}_{\epsilon}$  such as
$ {B}_{\epsilon} :=\{ ~\textbf{x}~ | ~\|\textbf{x}\| < \epsilon \} $.
Note that on the straight roads $\varphi$ can be considered as $\varphi=0$, while $\varphi \neq0$ on the curved roads.
For simplicity of analysis, we begin with considering a situation in which $\varphi =0$.

With the polytopic control setup and the feedback gain $K^{[p]}$ of each vertex, the optimal control gain for the reconstructed parameter $\hat{\boldsymbol{\theta}}$ takes the form
\begin{equation}
 K(\boldsymbol{\hat\theta})=\sum_{p=1}^{m+1} K^{[p]}\xi_{p}(\boldsymbol{\hat\theta}).
\label{eq:LPVcontrol}
\end{equation}
Note that controller (\ref{eq:LPVcontrol}) and the system in (\ref{eq:lpvsys}) share the same weighting functions $\xi_p$.
Then, we define the closed-loop system from the expressions of system (\ref{eq:lpvsys}) and controller (\ref{eq:LPVcontrol}) as:
\begin{equation}
\dot{\textbf{x}} = \sum_{p=1}^{m+1} \sum_{q=1}^{m+1} \xi_p(\boldsymbol{\hat\theta}) \xi_q(\boldsymbol{\hat\theta})
(A^{[p]}+B^{[p]}K^{[q]}) \textbf{x}
\label{eq:LPVcontrol}
\end{equation}

\begin{theorem}
Suppose that the Assumption~\ref{asum: PCA error} is satisfied and a convex set is given as
\begin{align*}
\Omega := \{G(\hat{\boldsymbol{\theta}})|G(\hat{\boldsymbol{\theta}})
=\sum_{p=1}^{m+1}\xi_p(\boldsymbol{\hat\theta}) {G}^{[p]},
\xi_p \geq 0,
\sum_{p=1}^{m+1} \xi_p(\boldsymbol{\hat\theta})=1\}
\end{align*}
where ${G}^{[p]}$ is the system at each vertex, ${G}(\boldsymbol{\hat\theta}_{v_p}^H)$.
Furthermore, supposed that given $\epsilon>0$, there exists $\gamma>0$ such that
$
\| g(t,\textbf{x}, \boldsymbol{\theta}) \| \leq \gamma \| \textbf{x} \|, ~\textbf{x} \in{B_\epsilon}$.
Then, the polytopic system (\ref{eq:continuous-time model2}) is locally exponentially stable for all $G(\hat{\boldsymbol{\theta}})\in \Omega$ and $\textbf{x} \in B_{\epsilon}$,
if there exist $0<P=P^T\in\mathbb{R}^{n_{\textbf{x}}\times n_{\textbf{x}}}$, $\alpha > 0$, $\gamma>0$, and $\tau \geq 0$ such that the following LMI conditions hold~\cite{lipolytopic}:
\begin{equation}
\begin{split}
&\Phi_{pp} \prec 0,~~~p \in  \{ 1, \cdots , m+1 \} \\
&\Phi_{pq} + \Phi_{qp} \prec 0,~~~p,q \in  \{ 1, \cdots , m+1 \},~p < q \\
\end{split}
\label{eq:LMI_condition}
\end{equation}
where
\begin{equation}
\Phi_{pq} =
\begin{bmatrix}
{\tilde{A}}^{[pq]^T} P  + P {\tilde{A}}^{[pq]} + \alpha P +\tau \gamma ^ 2 I  &P \\
(*)  & -\gamma I
\end{bmatrix}
\label{eq:LMI_vertex}
\end{equation}
with
\begin{equation}
{\tilde{A}}^{[pq]} = A^{[p]} + B^{[p]}K^{[q]}
\end{equation}

\end{theorem}

\begin{proof}
Once the state feedback controller $\textbf{u}=K(\hat{\boldsymbol{\theta}}) \textbf{x}$ is designed,
the system~(\ref{eq:continuous-time model2}) is rewritten as
\begin{equation}
\dot{\textbf{x}}
= \tilde{A}(\hat{\boldsymbol{\theta}}) \textbf{x} + g(t,\textbf{x}, \boldsymbol{\theta})
\label{eq:perturbed system}
\end{equation}
where
\begin{equation*}
\begin{split}
\tilde{A}(\hat{\boldsymbol{\theta}}) &=  A(\hat{\boldsymbol{\theta}}) + B(\hat{\boldsymbol{\theta}})K(\hat{\boldsymbol{\theta}}), \\
g(t,\textbf{x}, \boldsymbol{\theta}) &= \Delta A(\boldsymbol{\theta}) \textbf{x} + \Delta B(\boldsymbol{\theta}) K(\boldsymbol{\hat\theta}) \textbf{x}.
\end{split}
\end{equation*}
Here, $g(t,\textbf{x}, \boldsymbol{\theta})$ is the perturbation term result from reconstruction error or uncertainties.
From the assumption
$
\| g(t,\textbf{x}, \boldsymbol{\theta}) \| \leq \gamma \| \textbf{x} \|, ~\textbf{x} \in{B_\epsilon}
$ with some $\gamma>0$
under Assumption~\ref{asum: PCA error}.
The system~(\ref{eq:perturbed system}) is locally exponentially stable if that with a quadratic Lyapunov function $V(\textbf{x})=\textbf{x}^T P \textbf{x}$ and given $\alpha>0$, there exists $P>0$ and $\dot{V}(\textbf{x}) \leq -\alpha V(\textbf{x})$ for all $\textbf{x}$, such that:
\begin{equation*}
\begin{split}
\dot{V}(\textbf{x})+\alpha V(\textbf{x}) &= 2x^TP(\tilde{A}(\hat{\boldsymbol{\theta}}) x+g(t,\textbf{x}, \boldsymbol{\theta}))+\alpha \textbf{x}^TP\textbf{x}\\
&=
\begin{bmatrix}
\textbf{x}\\g(t,\textbf{x}, \boldsymbol{\theta})
\end{bmatrix}^T
\Lambda
\begin{bmatrix}
\textbf{x}\\g(t,\textbf{x}, \boldsymbol{\theta})
\end{bmatrix}\leq 0
\end{split}
\end{equation*}
where
$
\Lambda=
\begin{bmatrix}
\tilde{A}(\hat{\boldsymbol{\theta}}) ^TP+P\tilde{A}(\hat{\boldsymbol{\theta}}) +\alpha P & P\\
P&0
\end{bmatrix}
$
whenever
\begin{equation*}
\begin{bmatrix}
\textbf{x}\\g(t,\textbf{x}, \boldsymbol{\theta})
\end{bmatrix}^T
\begin{bmatrix}
\gamma^2 I & 0\\
0&-I
\end{bmatrix}
\begin{bmatrix}
\textbf{x}\\g(t,\textbf{x}, \boldsymbol{\theta})
\end{bmatrix}\geq 0.
\end{equation*}
Applying $\mathcal{S}$-procedure~\cite{boyd1994linear} leads to an inequality
\begin{equation*}
- \Lambda
\geq
\tau
\begin{bmatrix}
\gamma^2 I & 0\\
0&-I
\end{bmatrix}
\end{equation*}
for some $\tau \geq 0$. Thus, necessary and sufficient conditions for the existence of quadratic Lyapunov function $V(\textbf{x})$ can be expressed as LMI condition:
\begin{equation}
P>0, ~~
-
\begin{bmatrix}
\tilde{A}(\hat{\boldsymbol{\theta}}) ^TP+P\tilde{A}(\hat{\boldsymbol{\theta}}) +\alpha P + \tau\gamma^2 I & P\\
P&-\tau I
\label{eq:LPVLMI1}
\end{bmatrix}
\leq 0.
\end{equation}
Applying LMI condition (\ref{eq:LPVLMI1}) to the reconstructed system model $G^{[q]}$ at each vertex, it follows that:
\begin{equation}
\sum_{p=1}^{m+1} \sum_{q=1}^{m+1} \xi_p(\boldsymbol{\hat\theta}) \xi_q(\boldsymbol{\hat\theta}) \Phi_{pq} \prec 0
\label{eq:LMI_vertex_2}
\end{equation}
with $\Phi_{pq}$ defined in (\ref{eq:LMI_vertex}).
Then, from the complexity property, it is clear that condition (\ref{eq:LMI_condition}) guarantee (\ref{eq:LMI_vertex_2}).
Therefore, the closed system~(\ref{eq:perturbed system}) is locally exponentially stable.
\end{proof}
%

%
%
%
%
\begin{definition}
A system is robustly stable if there exists the system states $\textbf{x}$ remains bounded for all bounded disturbance~$g$.
\end{definition}
\begin{lemma}
Given $\alpha>0$, the polytopic system (\ref{eq:continuous-time model2}) is robustly stable, in presence of bounded disturbances, if
\begin{equation}
{\tilde{A}}^{[pq]^T} P  + P {\tilde{A}}^{[pq]} + \alpha P < 0
\label{eq:robuststable}
\end{equation}
for all $p,q = 1,...,m+1$.
\end{lemma}
The positive matrix $P$ which satisfies the condition in (\ref{eq:robuststable}) can be found by (\ref{eq:LPVLMI1}).
\begin{definition}
The set $\mathcal{D} \subseteq \mathbb{R}^{n_{\textbf{x}}}$ is a robustly positively invariant (RPI) set for the polytopic system (\ref{eq:continuous-time model2}) with perturbation $g$ in the bounded set $\Omega_g$, if with any ${\textbf{x}(t_0)} \in \mathcal{D}$, and any $g \in \Omega_g$, for all $t > t_0 $, it holds that ${\textbf{x}(t)} \in \mathcal{D}$.
\end{definition}
\begin{remark}
The polytopic system (\ref{eq:continuous-time model2}) is uniformly ultimately bounded to the set $\mathcal{D}$, if for each ${\textbf{x}}(t_0) \in \mathbb{R}^{n_{\textbf{x}}}$  there exists $T \geq 0$ such that any state trajectory of the system with initial condition ${\textbf{x}}(t_0)$ and any $g \in \Omega_g$ satisfies ${\textbf{x}}(t) \in \mathcal{D}$ for all $t \geq t_0 + T$.
\end{remark}
Note that on curved roads, $g(t,\textbf{x}, \boldsymbol{\theta})$ no longer satisfies Lipschitz condition since that
\begin{equation*}
\begin{split}
\| g(t,\textbf{x}, \boldsymbol{\theta})\|
& \leq \| \Delta A(\boldsymbol{\theta}) \textbf{x} + \Delta B(\boldsymbol{\theta}) K(\boldsymbol{\hat\theta}) \textbf{x}+
B_{\varphi}(\boldsymbol{\theta})\varphi\|\\
& \leq \gamma \| \textbf{x} \| + \|B_{\varphi}(\boldsymbol{\theta})\|\|\varphi\|.
\end{split}
\end{equation*}
In this case, robust stability of the system~(\ref{eq:continuous-time model2}) can be guaranteed by~(\ref{eq:LPVcontrol}) with bounded~$\|\varphi\|$.
With a common matrix $P$ such that (\ref{eq:robuststable}) holds for $G^{[q]}$ at each vertex, the uniformly ultimate boundedness of the state $\textbf{x}$ can be guaranteed, and the corresponding robustly positively invariant set $\mathcal{D}$ can be found. Thus on curved roads with bounded disturbance $\varphi \neq 0$, robust stability can be guaranteed. For further details of invariant-set computation, readers are referred to~\cite{martinez2013multisensor}.
\begin{figure}[t]
\includegraphics[width=1.23\textwidth]{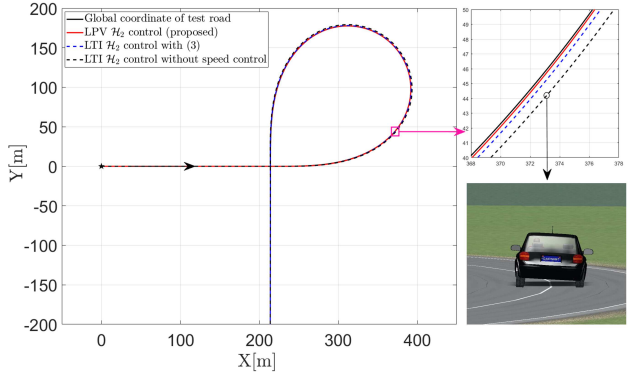}
\caption{Trajectory of the test road in the global coordinate.}
\label{fig:global}
\end{figure}
%
%
\begin{figure}[t]
\includegraphics[width=1\textwidth]{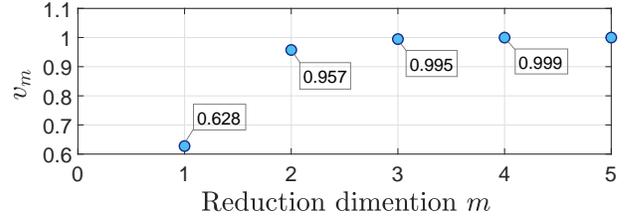}\caption{Fraction of total variation $v_m$ versus reduction dimension $m$. }
\label{fig:fraction}
\end{figure}

\section{Application Results}
The proposed system is validated by MATLAB/Simulink with a vehicle plant by CarSim.
We chose CarSim for its accurate modeling for vehicle dynamics and wide application in industries.
Simulation for PCA data collection is set at various flat highway roads.
The road surface condition is chosen to be dry asphalt to ensure the linear relationship between lateral force $F_y$ and tire slip angle $\alpha$.
The test road is a zero-bank interchange highway road with a curved section radius of $80$~[m] and dry asphalt road condition.
The overall trajectory of the test road is shown in Fig.~\ref{fig:global}.

Here we chose the reduction dimension $m=3$.
The original scheduling variables, $\theta_{i}$, are all combinations out of the $3$ varying parameters $V_x$, $C_{\alpha f}$ and $C_{\alpha r}$.
By looking into the definition of $\boldsymbol{\theta}$, one can find that any $3$ out of $\theta_{i}$ cover the variations of $V_x$, $C_{\alpha f}$ and $C_{\alpha r}$.
Thus we consider that by choosing $m=3$, both accuracy and simplicity of the model can be guaranteed.
To measure the approximation quality, the fraction of total variation $v_m$ is determined by the singular values $\sigma_i$ in $\Sigma_s$ and $\Sigma_n$ as $v_m = \frac{\sum_{i=1}^{m} \sigma^2_i}{\sum_{i=1}^{l} \sigma^2_i}$ , and compared with different choice of reduced dimension in Fig.~\ref{fig:fraction}.
The trajectories of original scheduling variables $\theta_i$ and their approximations $\hat{\theta}_i$ by choosing $m=3$ are compared in Fig.~\ref{fig:approx}.
From the comparison, one can conclude that the trajectories are matched well, and accuracy is ensured with the reduction dimension $m=3$.
Notice that the choice of $m$ changes with different driving conditions.

The PCA dimension reduction procedure reduced all the candidate vertices from $L=2^5=32$ to $L=2^3=8$.
With $l$ scheduling variables, the $l+1$ number of vertices is considered sufficient to construct a convex polytope.
Then the number of possible polytopic model candidates is reduced from $_{2^5}C_{5+1}=906192$ to $5$, thanks to the visuality of 3-dimensional polytopes as is shown in Fig.~\ref{fig:vertex}.
The computational complexity for both online and offline calculations is significantly reduced while the accuracy of LPV models is maintained.
\begin{figure}[t]
\centering
\SetFigLayout{3}{2}
  \subfigure[$\theta_1=V_x$]{\includegraphics{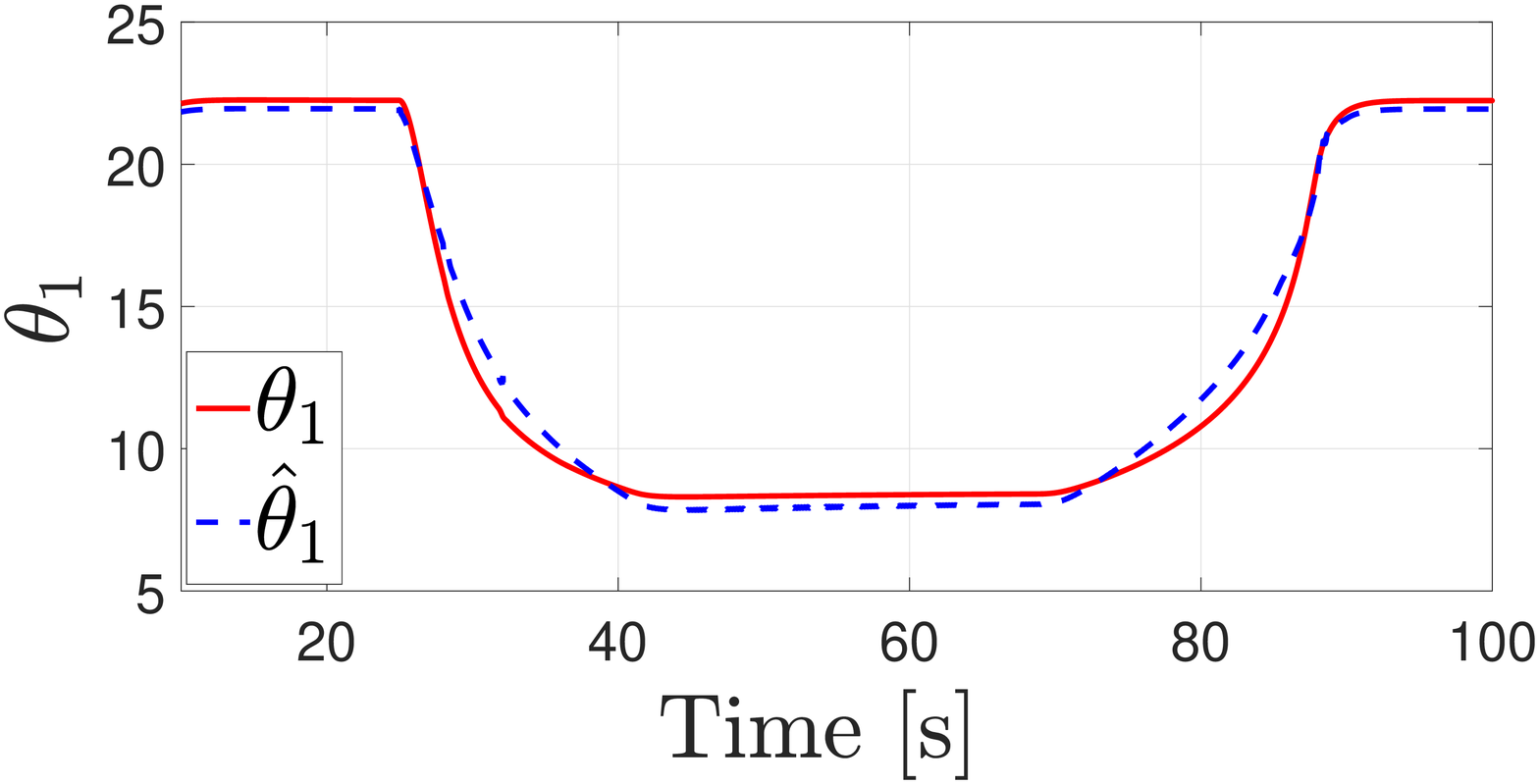}}
  \hfill
  \vspace{-0.1cm}
  \subfigure[$\theta_2 = 2C_{\alpha f}$]{\includegraphics{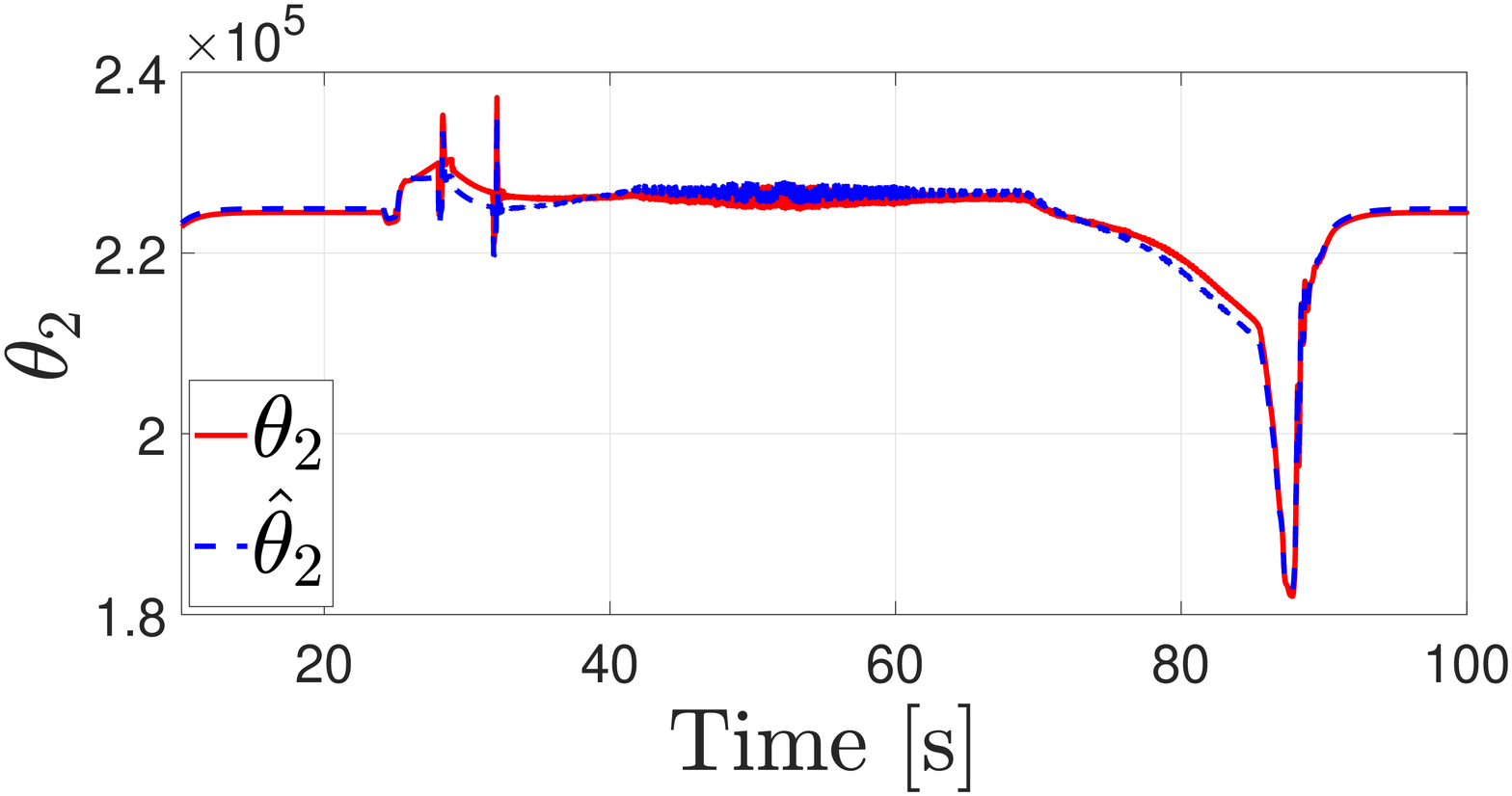}}
  \hfill
  \vspace{-0.1cm}
  \subfigure[$\theta_3 = 2C_{\alpha f}/V_x$]{\includegraphics{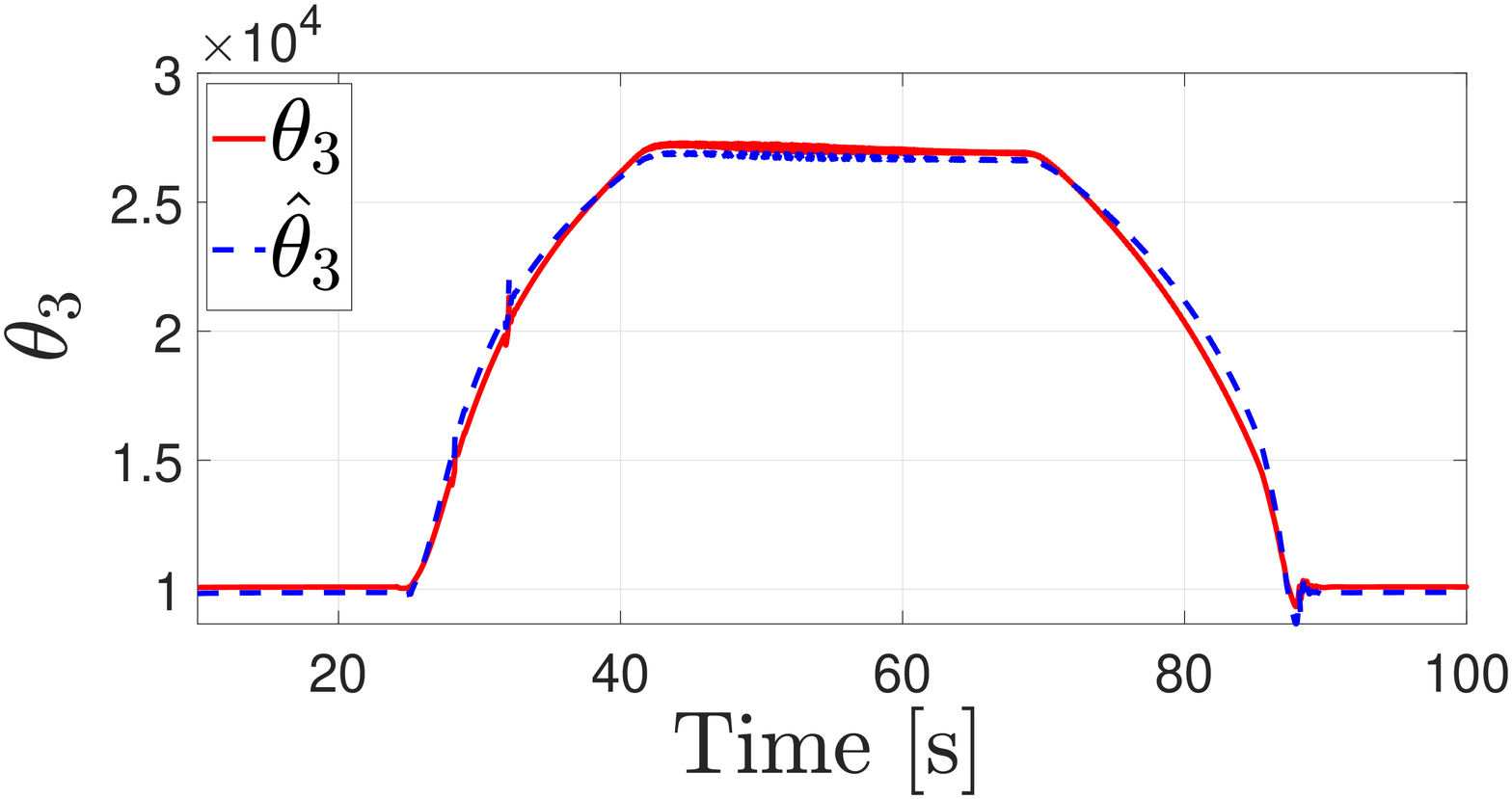}}
  \hfill
  \vspace{-0.1cm}
  \subfigure[$\theta_4 = 2C_{\alpha r}$]{\includegraphics{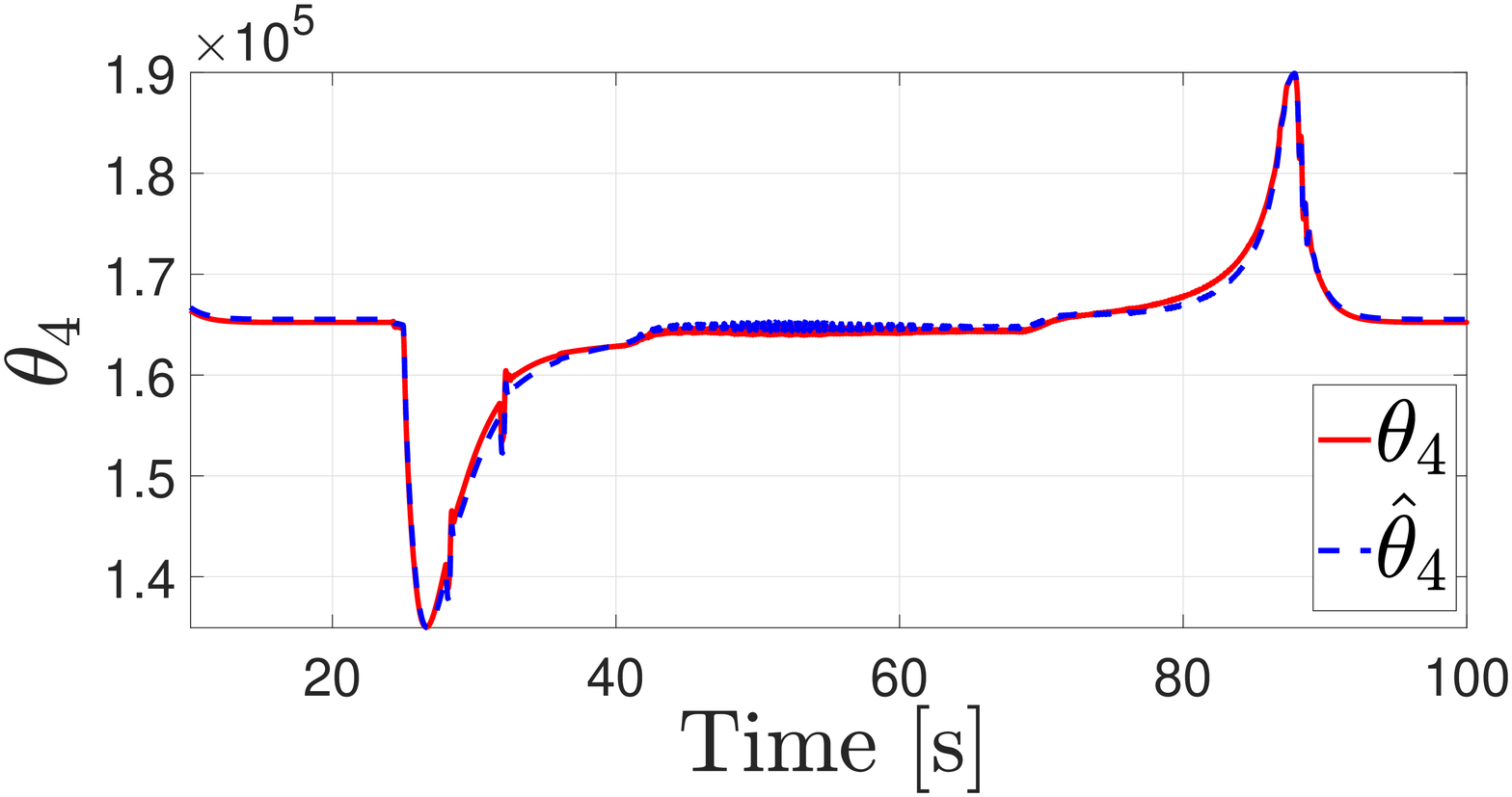}}
  \hfill
  \vspace{-0.05cm}
  \subfigure[$\theta_5 = 2C_{\alpha r}/V_x$]{\includegraphics{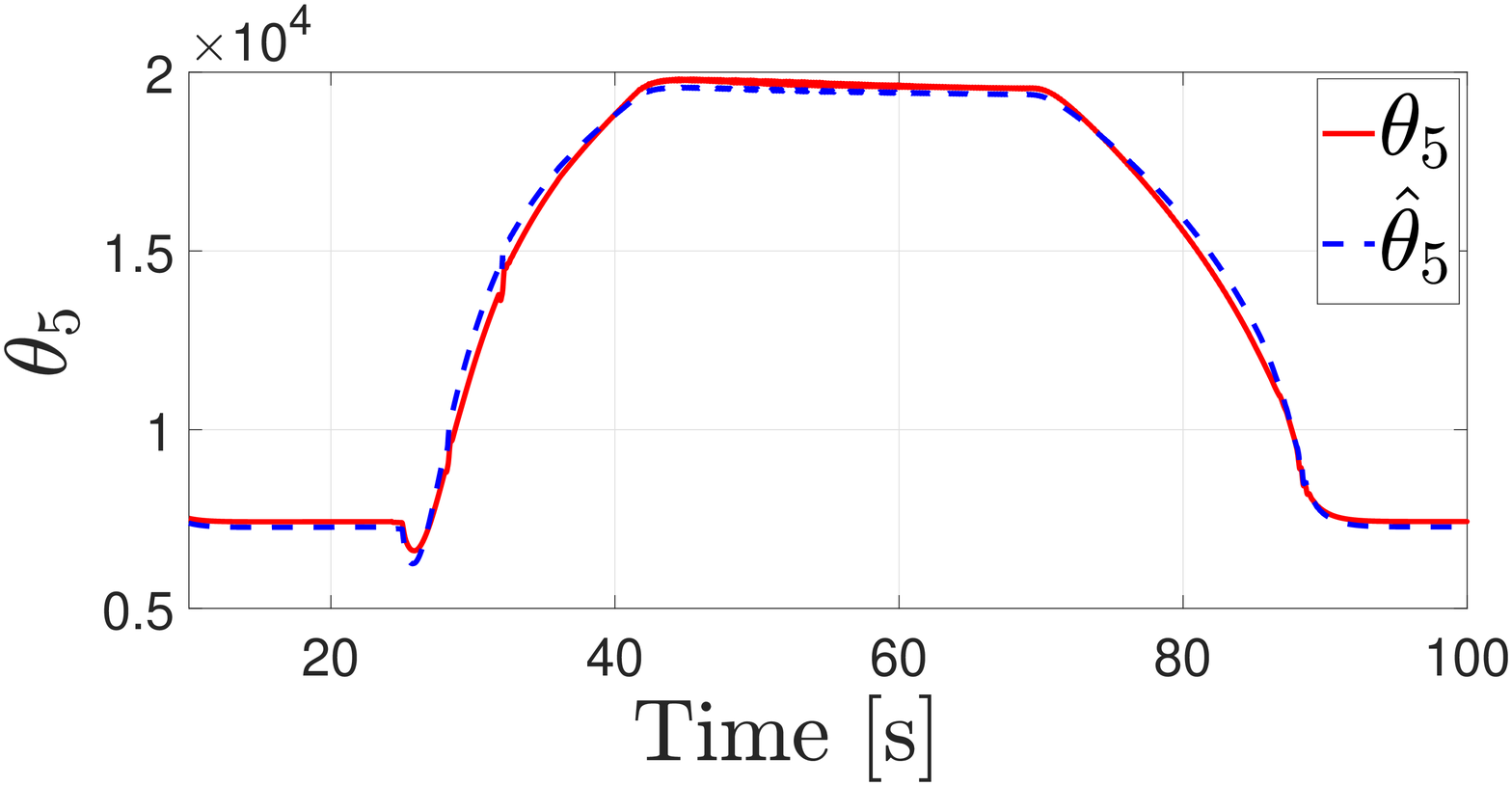}}
  \hfill
\caption{Scheduling variable $\theta_i$(red line) and approximation $\hat{\theta}_i$ (blue line) with $m=3$.}
\label{fig:approx}
\vspace{-0.2cm}
\end{figure}

We controlled the vehicle longitudinal speed, $V_x$ by commanding the desired longitudinal vehicle speed, $V_x^d$ satisfying~(\ref{eq:desired speed}) with a simple PD controller as depicted in~Fig.~\ref{fig:simulation}.
The varying cornering stiffness along vehicle roll motion is estimated with the measurements in CarSim.
The comparison of the steering angle and lane-tracking control performance generated by the LTI controller and the proposed LPV controller are shown in Fig.~\ref{fig:ey} and Fig.~\ref{fig:steer} respectively.
%
%
The longitudinal speed in the LTI controller design is fixed as $V_x=50$~[km/h], and cornering stiffness is set as the mean value of the collected cornering stiffness trajectories.
Compared to the controller for the LTI model, the proposed controller generates a better tracking performance with reduced lateral offset error $e_y$ as is shown in Fig.~\ref{fig:ey}.
The difference between $e_y$ generated by the LTI controller and the proposed controller is caused by adjusting  the steering angle three times indicated as $P1$, $P2$, and $P3$ in Fig.~\ref{fig:steer}.
%
For this reason, the steady-state lateral offset error $e_y$ is reduced, even though there is no much difference between the steering angles generated by the proposed controller and the LTI controller at the steady-state period.
To emphasize the necessity of vehicle longitudinal speed control in the curved section, we observed the vehicle roll angle is largely reduced, as is shown in Fig.~\ref{fig:roll}.
Also the global coordinates generated by both controllers are compared in Fig.~\ref{fig:global}, where we can observed from CarSim that the vehicle tends to have a large roll angle and cross the line in the beginning of the curved section.
Therefore, we stress it again that variations of longitudinal speed and cornering stiffness should not be ignored in lateral control design in curved roads.
%
%
\begin{figure}[t]
\hspace{.5cm}
\includegraphics[width=1.95\textwidth]{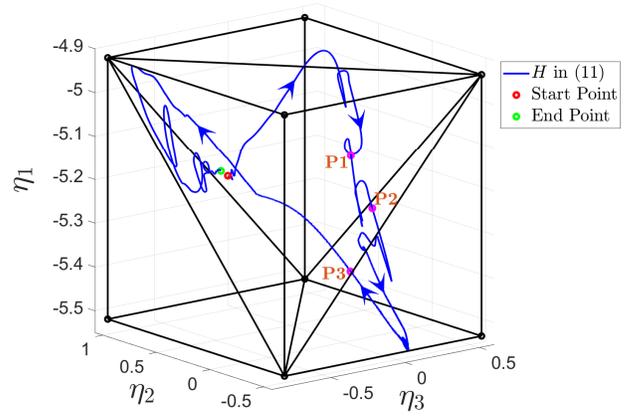}\caption{Trajectory of the reduced-dimension scheduling variable and the corresponding polytopic.}
\label{fig:vertex}
\vspace{-0.2cm}
\end{figure}
%
%
\begin{figure}[t]
\includegraphics[width=1\textwidth]{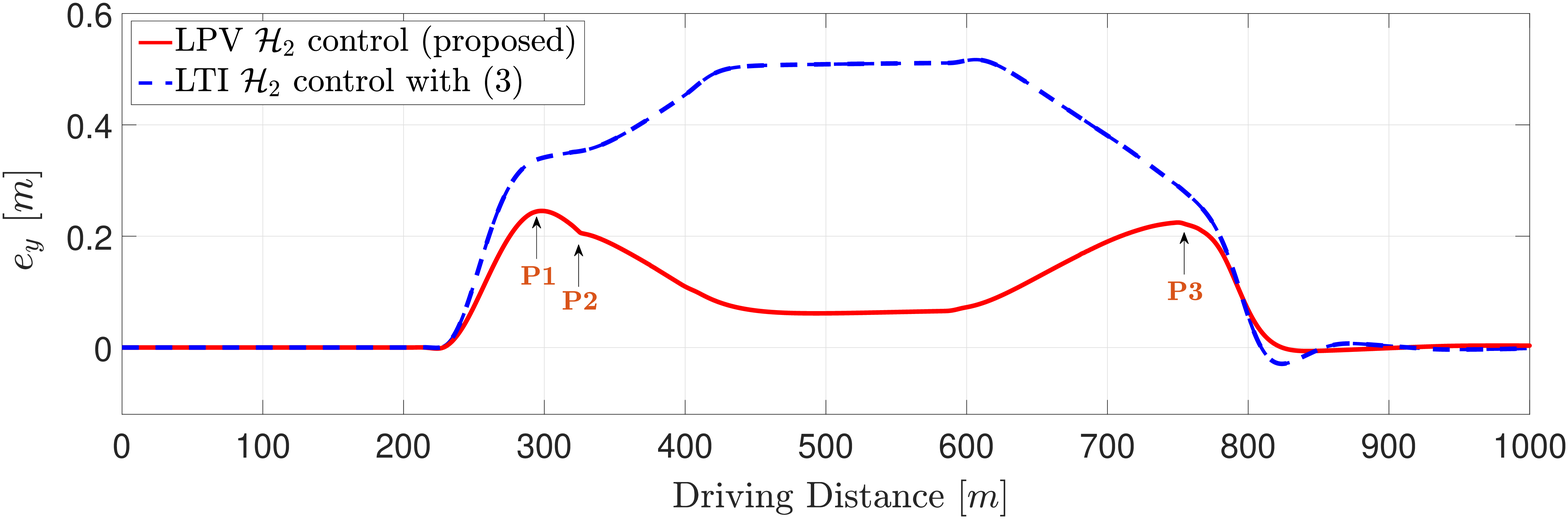}
\caption{Comparison of lateral offset error generated by LTI and LPV lateral control with longitudinal control.}
\label{fig:ey}
\vspace{-0.2cm}
\end{figure}
%
%
\begin{figure}[t]
\includegraphics[width=1\textwidth]{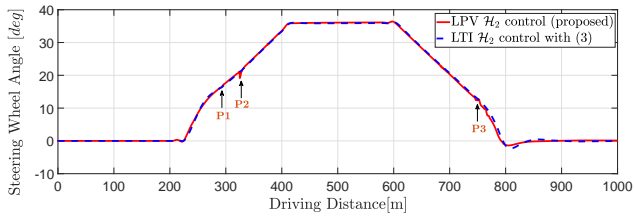}
\caption{Comparison of steering angle generated by LTI and LPV lateral control with longitudinal control.}
\label{fig:steer}
\end{figure}
\begin{figure}[t]
\includegraphics[width=1\textwidth]{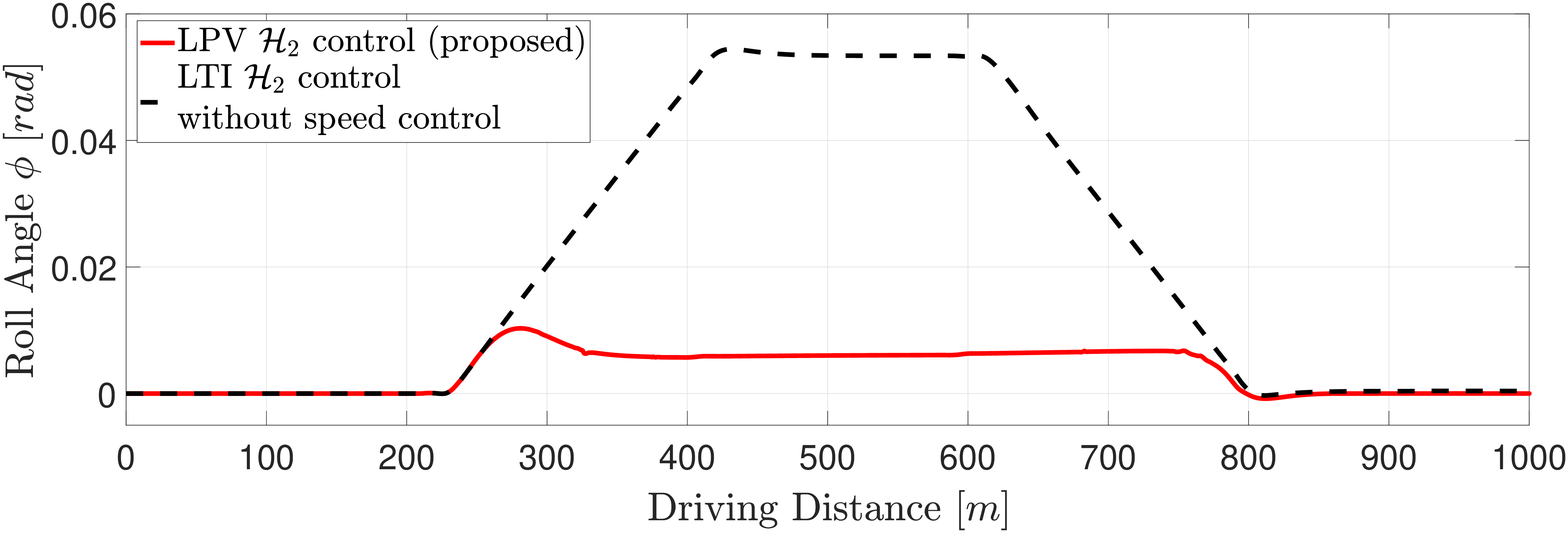}
\caption{Comparison of roll angle generated by the proposed LPV lateral control with longitudinal control and LTI lateral control without longitudinal control.}
\label{fig:roll}
\end{figure}
%
\section*{ACKNOWLEDGEMENT}
This work was supported by the National Research Foundation of Korea(NRF) grant funded by the Korea government(MSIT) (No. 2021R1A2C2009908, Data-Driven Optimized Autonomous Driving Technology Using Open Set Classification Method).

\section{CONCLUSION}

In this paper, we proposed the gain scheduling method with the reduction of varying parameters. Longitudinal speed control was necessary on the interchange road for reducing the roll angle. The varying vehicle speed and roll motion impact the parameters of the lateral dynamic model, such as tire cornering stiffness. Thus, the LPV model was used to consider the parameter variations. The problem of considering the multiple parameters is the high number of scheduling variables and the resulting computational complexity. The PCA-based parameter reduction was exploited to remedy the problem.
From the simulation results with CarSim, we validated the effectiveness of the proposed algorithm for lane-keeping maneuvering on the interchange road. It is expected that the proposed system is suitable to improve the performance of lateral control in autonomous driving.
%


\bibliographystyle{IEEEtran}        

\begin{thebibliography}{10}
\providecommand{\url}[1]{#1}
\csname url@samestyle\endcsname
\providecommand{\newblock}{\relax}
\providecommand{\bibinfo}[2]{#2}
\providecommand{\BIBentrySTDinterwordspacing}{\spaceskip=0pt\relax}
\providecommand{\BIBentryALTinterwordstretchfactor}{4}
\providecommand{\BIBentryALTinterwordspacing}{\spaceskip=\fontdimen2\font plus
\BIBentryALTinterwordstretchfactor\fontdimen3\font minus
  \fontdimen4\font\relax}
\providecommand{\BIBforeignlanguage}[2]{{%
\expandafter\ifx\csname l@#1\endcsname\relax
\typeout{** WARNING: IEEEtran.bst: No hyphenation pattern has been}%
\typeout{** loaded for the language `#1'. Using the pattern for}%
\typeout{** the default language instead.}%
\else
\language=\csname l@#1\endcsname
\fi
#2}}
\providecommand{\BIBdecl}{\relax}
\BIBdecl

\bibitem{chen2001differential}
B.-C. Chen and H.~Peng, ``Differential-braking-based rollover prevention for
  sport utility vehicles with human-in-the-loop evaluations,'' \emph{Vehicle
  system dynamics}, vol.~36, no. 4-5, pp. 359--389, 2001.

\bibitem{yoon2008unified}
J.~Yoon, W.~Cho, B.~Koo, and K.~Yi, ``Unified chassis control for rollover
  prevention and lateral stability,'' \emph{IEEE Trans. on Vehicular
  Technology}, vol.~58, no.~2, pp. 596--609, 2008.

\bibitem{rajamani2011vehicle}
R.~Rajamani, \emph{Vehicle dynamics and control}.\hskip 1em plus 0.5em minus
  0.4em\relax Springer Science \& Business Media, 2011.

\bibitem{kim2019vehicular}
S.~W. Kim, Y.~W. Jeong, J.~S. Kim, S.-H. Lee, and C.~C. Chung, ``Vehicular
  vertical tire forces estimation using unscented kalman filter,'' in
  \emph{2019 12th ASCC}.\hskip 1em plus 0.5em minus 0.4em\relax IEEE, 2019, pp.
  325--330.


\bibitem{sierra2006cornering}
C.~Sierra, E.~Tseng, A.~Jain, and H.~Peng, ``Cornering stiffness estimation
  based on vehicle lateral dynamics,'' \emph{Vehicle System Dynamics}, vol.~44,
  no. sup1, pp. 24--38, 2006.

\bibitem{wang2013tire}
R.~Wang and J.~Wang, ``Tire--road friction coefficient and tire cornering
  stiffness estimation based on longitudinal tire force difference
  generation,'' \emph{CEP}, vol.~21, no.~1, pp. 65--75, 2013.

\bibitem{berntorp2018tire}
K.~Berntorp and S.~Di~Cairano, ``Tire-stiffness and vehicle-state estimation
  based on noise-adaptive particle filtering,'' \emph{IEEE Transactions on
  Control Systems Technology}, vol.~27, no.~3, pp. 1100--1114, 2018.



\bibitem{lee2020autonomous}
S.-H. Lee and C.~C. Chung, ``Autonomous-driving vehicle control with composite
  velocity profile planning,'' \emph{IEEE Tran. on Control Systems
  Technology}, 2020.

\bibitem{li2014lpv}
M.~Li, Y.~Jia, and J.~Du, ``Lpv control with decoupling performance of 4ws
  vehicles under velocity-varying motion,'' \emph{IEEE Trans. on Control
  Systems Technology}, vol.~22, no.~5, pp. 1708--1724, 2014.

\bibitem{kwiatkowski2008pca}
A~Kwiatkowski and H.~Werner, ``Pca-based parameter set mappings for lpv models
  with fewer parameters and less overbounding,'' \emph{IEEE Trans. on
  Control Systems Technology}, vol.~16, no.~4, pp. 781--788, 2008.

\bibitem{lipolytopic}
Li, Panshuo and Nguyen, Anh-Tu and Du, Haiping and Wang, Yan and Zhang, Hui, ``Polytopic LPV Approaches for Intelligent Automotive Systems: State of the Art and Future Challenges.''

\bibitem{hashemi2009lpv}
S.~M. Hashemi, H.~S. Abbas, and H.~Werner, ``Lpv modelling and control of a
  2-dof robotic manipulator using pca-based parameter set mapping,'' in
  \emph{Proc. of the 48h IEEE CDC}.\hskip
  1em plus 0.5em minus 0.4em\relax IEEE, 2009, pp. 7418--7423.

\bibitem{rizvi2016kernel}
S.~Z. Rizvi, J.~Mohammadpour, R.~T{\'o}th, and N.~Meskin, ``A kernel-based pca
  approach to model reduction of linear parameter-varying systems,'' \emph{IEEE
  Trans. on Control Systems Technology}, vol.~24, no.~5, pp. 1883--1891,
  2016.

\bibitem{rizvi2018model}
S.~Z. Rizvi, F.~Abbasi, and J.~M. Velni, ``Model reduction in linear
  parameter-varying models using autoencoder neural networks,'' in \emph{2018
  Annual ACC}.\hskip 1em plus 0.5em minus 0.4em\relax IEEE, 2018, pp.
  6415--6420.

\bibitem{koelewijn2020scheduling}
P.~Koelewijn and R.~T{\'o}th, ``Scheduling dimension reduction of lpv models-a
  deep neural network approach,'' in \emph{Proc. of 2020 ACC}.\hskip 1em plus 0.5em
  minus 0.4em\relax IEEE, 2020, pp. 1111--1117.

\bibitem{jackson2005user}
J.~E. Jackson, \emph{A user's guide to principal components}.\hskip 1em plus
  0.5em minus 0.4em\relax John Wiley \& Sons, 2005, vol. 587.

\bibitem{son2014robust}
Y.~S. Son, W.~Kim, S.-H. Lee, and C.~C. Chung, ``Robust multirate control
  scheme with predictive virtual lanes for lane-keeping system of autonomous
  highway driving,'' \emph{IEEE Trans. on Vehicular Technology}, vol.~64,
  no.~8, pp. 3378--3391, 2014.

\bibitem{lee2016robust}
S.-H. Lee and C.~C. Chung, ``Robust multirate on-road vehicle localization for
  autonomous highway driving vehicles,'' \emph{IEEE Trans. on Control
  Systems Technology}, vol.~25, no.~2, pp. 577--589, 2016.

\bibitem{reiss2020nonasymptotic}
M.~Rei{\ss}, M.~Wahl \emph{et~al.}, ``Nonasymptotic upper bounds for the
  reconstruction error of pca,'' \emph{Annals of Statistics}, vol.~48, no.~2,
  pp. 1098--1123, 2020.

\bibitem{milbradt2020high}
C.~Milbradt and M.~Wahl, ``High-probability bounds for the reconstruction error
  of pca,'' \emph{Statistics \& Probability Letters}, vol. 161, p. 108741,
  2020.

\bibitem{boyd1994linear}
S.~Boyd, L.~El~Ghaoui, E.~Feron, and V.~Balakrishnan, \emph{Linear matrix
  inequalities in system and control theory}.\hskip 1em plus 0.5em minus
  0.4em\relax SIAM, 1994.

\bibitem{khalil2014nonlinear}
Khalil, Hassan K, \emph{Nonlinear control}.\hskip 1em plus 0.5em minus
  0.4em\relax Pearson Higher Ed, 2014.

\bibitem{martinez2013multisensor}
Martinez, John J and Varrier, S{\'e}bastien, ``Multisensor Fault-Tolerant Automotive Control,'' \emph{Robust Control and Linear Parameter Varying Approaches}, p. 267-287.\hskip 1em plus 0.5em minus
  0.4em\relax Springer, 2013.
  2020.

\end{thebibliography}


\end{document}